\documentclass[letterpaper, 10 pt, conference]{ieeeconf}

\usepackage[noadjust]{cite}

\IEEEoverridecommandlockouts                               
\let\proof\relax

\let\labelindent\relax

\usepackage{amsmath}
\usepackage{amssymb}
\usepackage{amsthm}
\usepackage{comment}
\usepackage{mathrsfs}
\usepackage{mathtools}
\usepackage{enumitem}
\usepackage{float,color}
\usepackage[dvipsnames]{xcolor}
\usepackage[ruled,linesnumbered,vlined]{algorithm2e}
\usepackage{soul}
\usepackage{graphicx}
\usepackage{epstopdf}
\usepackage{algorithmic}
\usepackage{textcomp}

\newcommand{\emnote}[1]{{\color{blue}[Erik -- #1]}}

\newcommand{\kzedit}[1]{{\color{magenta}#1}}

\newcommand{\iPi}{\mathit{\Pi}}
\newcommand{\PP}{\mathbb{P}}
\newcommand{\FF}{\mathbb{F}}
\newcommand{\EE}{\mathbb{E}}
\newcommand{\Fs}{\mathcal{F}}
\newcommand{\Ms}{\mathcal{M}}
\newcommand{\Ns}{\mathcal{N}}

\newcommand{\Ts}{\mathcal{T}}

\newcommand{\Os}{\mathcal{O}}
\newcommand{\Zs}{\mathcal{Z}}
\newcommand{\Es}{\mathcal{E}}

\newcommand{\RR}{\mathbb{R}}
\DeclareMathOperator{\tr}{Tr}
\DeclareMathOperator{\svec}{svec}
\DeclareMathOperator*{\argmin}{argmin}

\newtheorem{theorem}{Theorem}
\newtheorem*{theorem*}{Theorem}
\newtheorem{lemma}{Lemma}
\newtheorem{assumption}{Assumption}
\newtheorem{definition}{Definition}

\newtheorem*{proposition*}{Proposition}

\makeatletter
\def\old@comma{,}
\catcode`\,=13
\def,{%
  \ifmmode%
    \old@comma\discretionary{}{}{}%
  \else%
    \old@comma%
  \fi%
}
\makeatother

\newtheorem{proposition}{Proposition}
\title{\LARGE \bf 
Reinforcement Learning in Non-Stationary  Discrete-Time \\
Linear-Quadratic Mean-Field Games
} 

\author{Muhammad~Aneeq~uz~Zaman, Kaiqing~Zhang, Erik~Miehling, and Tamer~Ba{\c s}ar
\thanks{The authors are affiliated with the Coordinated Science Laboratory, University of Illinois at Urbana--Champaign Urbana, IL 61801.}
\thanks{Research support in part by Grant FA9550-19-1-0353 from AFOSR, and in part by US Army Research Laboratory (ARL) Cooperative Agreement W911NF-17-2-0196.}
}

\begin{document}
\maketitle
\thispagestyle{empty}
\pagestyle{empty}


\begin{abstract}
In this paper, we study large population multi-agent reinforcement learning (RL) in the context of discrete-time linear-quadratic mean-field games (LQ-MFGs). Our setting differs from  most existing work on RL for MFGs, in that we consider a  \emph{non-stationary} MFG over an infinite horizon. We propose an actor-critic algorithm to iteratively compute the \emph{mean-field equilibrium} (MFE) of the LQ-MFG. There are two primary challenges: i) the non-stationarity of the MFG induces a \emph{linear-quadratic tracking} problem, which requires solving a \emph{backwards-in-time (non-causal)} equation that cannot be solved by standard (causal) RL algorithms; ii) Many 
RL {algorithms} assume that the states are sampled from the stationary distribution of a Markov chain (MC), that is, the chain is already \textbf{mixed}, an assumption that is not satisfied for real data sources. We first identify  that the mean-field trajectory follows linear dynamics, allowing the problem to be reformulated as a linear quadratic Gaussian problem. Under this reformulation, we propose an actor-critic algorithm that allows samples to be drawn from an unmixed MC. Finite-sample convergence guarantees for the algorithm are then provided. To characterize the performance of our algorithm in multi-agent RL, we have developed an error bound with respect to the Nash equilibrium of the finite-population game.  
\end{abstract}

\section{Introduction}

\label{sec:intro}
%
%
%
%

Recent years have witnessed the tremendous progress of reinforcement learning (RL) \cite{lowe2017multi,zhang2018fully,zhang2018finite} and planning \cite{best2018dec,zhang2019online} in multi-agent settings
; see \cite{zhang2019multi} for  a recent overview of multi-agent RL (MARL). 
The primary challenge that MARL algorithms face is their scalability 
due to the exponential increase in complexity in the number of agents. 
This difficulty prevents the use of many MARL algorithms in real-world applications, \emph{e.g.}, \cite{couillet2012electrical,cardaliaguet2018mean}.  

To address this challenge, we focus on the 
framework of \emph{mean-field games} (MFGs)
, originally introduced in \cite{huang2003individual,lasry2007mean}. 
The core idea is that the interaction among a large population of agents is well-approximated by the aggregate behavior of the agents, or the \emph{mean-field trajectory}, where the influence of each agent has a negligible effect on the mass. 
Following the 
\emph{Nash certainty equivalence} (NCE) principle \cite{huang2006large}, the solution to an MFG, referred to as a \emph{mean-field equilibrium} (MFE), can be obtained 
by computing a best-response to some mean-field trajectory that is consistent with the aggregate behavior of all agents. This decouples the solution process into the computation of a best-response 
for a fixed mean-field trajectory
, and the update of the mean-field trajectory. Computation of the best-response 
can be done in a model-free fashion using single-agent RL techniques \cite{yang2019global}. 
The computed MFE provides a reasonably accurate approximation of the actual Nash Equilibrium (NE) of the corresponding finite-population dynamic game, a common model for  MARL \cite{huang2007large}. 
Due to this desired property, 
there have been a growing interest in studying  RL  algorithms in MFGs \cite{cardaliaguet2017learning,subramanian2019reinforcement,guo2019learning,elie2019approximate,fu2019actor}. 

Serving as a standard, but significant, benchmark for general MFGs, linear-quadratic MFGs (LQ-MFGs) \cite{huang2007large,bensoussan2016linear}, 
have received significant attention in the literature. Under this setting, the cost function describing deviations in the state from the mean-field state, as well as the magnitude of the control, are assumed to be quadratic while the transition dynamics are assumed to be linear. Intuitively, the cost causes 
each agent to \emph{track} the collective behavior of the population, which, for any fixed mean-field trajectory, leads to a \emph{linear-quadratic tracking} (LQT) subproblem for each agent. 
While most of the work has been done in the continuous-time setting \cite{huang2007large,bensoussan2016linear,huang2018linear}, the discrete-time counterpart, the focus of our paper, has received relatively less attention \cite{moon2014discrete}. 
  
Despite the existence of learning algorithms for specific classes of MFGs  \cite{cardaliaguet2017learning,subramanian2019reinforcement,guo2019learning,elie2019approximate,fu2019actor}, the current literature does not apply to the LQ-MFG setting; see the related work subsection 
for a complete comparison. Most relevant to our setting is the recent independent work of \cite{fu2019actor} 
in which each agent's subproblem, given any fixed mean-field trajectory, is treated as a linear quadratic regulator (LQR) with drift. This is possible due to the restriction to mean-field trajectories that are constant over time (referred to as \emph{stationary mean-fields} in the literature \cite{subramanian2019reinforcement}). 
This is in contrast to the LQT subproblems in the LQ-MFG literature \cite{huang2007large,bensoussan2016linear,huang2018linear} -- a more 
standard setup and one we follow in this paper. While the former admits a \emph{causal} optimal control that can be solved for using RL algorithms for LQR problems \cite{bradtke1993reinforcement,fazel2018global}, the latter leads to a \emph{non-causal} optimal control problem, 
which is well known to be challenging from a model-free perspective 
\cite{kiumarsi2014reinforcement,modares2014linear}. We 
present conditions such that the mean-field trajectory, of the MFE, follows linear dynamics. Hence, we can restrict attention to \emph{linear} mean-field trajectories, allowing for a \emph{causal} reformulation that enables the development of model-free RL algorithms. 

Furthermore, some recent RL algorithms for MFGs assume that data samples are drawn from the stationary distribution of a Markov chain (MC) under some policy \cite{fu2019actor}, and sometimes even done so independently \cite{anahtarci2019fitted}. Though facilitating analyses, data trajectories in practice are usually sampled from an \emph{unmixed} MC. 
Our analyses reflect this more realistic sampling scheme. 


\vspace{5pt}
\noindent\textbf{Contribution.}
We develop a provably convergent RL algorithm for \emph{non-stationary} and \emph{infinite-horizon} discrete-time LQ-MFGs, inspired by the formulations of \cite{huang2007large,moon2014discrete}. 
Our contribution is three-fold: ({\bf 1}) By identifying useful \emph{linearity} properties of the MFE, 
we develop an actor-critic algorithm that 
addresses the non-stationarity of the MFE; as opposed to \cite{fu2019actor,subramanian2019reinforcement}; 
({\bf 2}) We provide a finite-sample analysis of our actor-critic algorithm, under the more realistic sampling setting with unmixed Markovian state trajectories
; ({\bf 3}) We quantify the error bound  of our approximate MFE obtained from the algorithm, as an $\epsilon$-
NE of the original finite-population MARL problem. 



\vspace{5pt}
\noindent\textbf{Related Work.}
Rooted in the original MFG formulation \cite{huang2003individual,lasry2007mean,saldi2018markov}, LQ-MFGs have been proposed mostly for the continuous-time setting \cite{huang2007large,bensoussan2016linear,huang2018linear} and less so for the discrete-time setting \cite{moon2014discrete,uz2020approximate,fu2019actor}. Our previous work \cite{uz2020approximate} proposes an MFE approximation algorithm and does not study the linearity properties of the MFE. 
Recently, the work of \cite{fu2019actor} has also considered learning in discrete-time LQ-MFGs. However, the subproblem therein (given a fixed mean-field trajectory) is modeled as an LQR problem with drift.  This deviation from the convention \cite{huang2007large,bensoussan2016linear,huang2018linear} yields a problem that can be solved using RL algorithms for LQR problems. In particular, an actor-critic algorithm was developed in \cite{fu2019actor} to find the \emph{stationary} MFE. 

Beyond the LQ setting, there is a burgeoning interest in developing RL algorithms for MFGs \cite{yin2013learning,cardaliaguet2017learning,subramanian2019reinforcement,guo2019learning,elie2019approximate}. 
To emphasize the relationship between MFG and RL, most work \cite{subramanian2019reinforcement,guo2019learning,elie2019approximate} has studied the discrete-time setting.  
In particular, \cite{subramanian2019reinforcement,guo2019learning} develop both policy-gradient and Q-learning based algorithms, but with a focus on MFGs with a stationary 
MFE. In contrast,  \cite{elie2019approximate} is the first paper that considers \emph{non-stationary} MFEs. However, the results therein do not apply to the LQ-MFG model of the present paper, since \cite{elie2019approximate} considered finite horizons, and the state-action spaces, though continuous, are required to be convex and compact. 
More recently, \cite{anahtarci2019fitted} proposed a fitted-Q learning algorithm for MFGs, which learns a stationary MFE.  
In fact, as pointed out in \cite{subramanian2019reinforcement},  all prior work was restricted to either stationary MFGs or finite-horizon settings.




The remainder of the paper is organized as follows. 
In Section \ref{sec:model}, we introduce the LQ-MFG problem and discover useful linearity properties of the MFE, offering a characterization of the MFE. We then develop an actor-critic algorithm in Section \ref{sec:RL}, followed by the finite-sample and finite-population analyses in Section \ref{sec:Analysis}. Concluding remarks are provided in Section \ref{sec:Conc}. 

%

\section{Linear-Quadratic Mean-Field Games}
\label{sec:model}

Consider a dynamic game with $N<\infty$ agents playing on an infinite horizon. Each agent $n \in [N]$ is responsible for controlling its own state, denoted by $Z_t^n \in \RR^m$, via selection of control actions, denoted by $U_t^n \in \RR^p$. The state process corresponding to each agent $n$ evolves according to the following linear time-invariant (LTI) dynamics,\begin{align}\label{eq:finitesystem}
	Z_{t+1}^n = A Z_t^n +B U_t^n +W_t^n,
\end{align} 
with state matrix $A \in \RR^{m \times m}$, input matrix $B \in \RR^{m \times p}$, and noise terms $W_t^n$, $t=0,1,\ldots$, independently and identically distributed with Gaussian distribution $\Ns(0,\Sigma_w)$. The pair $(A,B)$ is assumed to be controllable. For each $n$, the initial state $Z_0^n$ is generated by distribution $\Ns(\nu_0,\Sigma_0)$. 
Each agent $n$'s initial state is assumed to be independent of the noise terms, $W_s^n$, $n\in[N]$, $s=1,2,\ldots$, and other agents' initial states, $Z_0^{n'}$, $n'\neq n$. At the beginning of each time step, each agent observes every other agent's state.\footnote{This is a game of full shared history. We will see later that actually full sharing of the state information is not needed, and with each agent accessing only its local state with no memory will be sufficient.} 
Thus, under perfect recall, the information of agent $n$ at time $t$ is $I_t^n = \big((Z_0^1,\ldots,Z_0^N),U_0^n;\ldots;(Z_{t-1}^1,\ldots,Z_{t-1}^N),U_{t-1}^n ;(Z_t^1,\ldots,Z_t^N)\big)$. 
A control policy for agent $n$ at time $t$, denoted by $\pi_t^n$, maps its information $I_t^n$ to a control action $U_t^n\in\RR^p$. The sequence of control policies for agent $n$ is called a control law $\pi^n := (\pi^n_0, \pi^n_1,\ldots$) with the set of all control laws denoted by $\iPi$. 
The joint control law is the collection of control laws over all $n$, denoted by $\pi = (\pi^1,\ldots,\pi^N)$. The agents are coupled via their expected cost functions, which penalizes both the control magnitude and the deviation of each agent's state from the average state. The expected cost for agent $n$ under joint control law $\pi$, denoted by $J_n^N(\pi^n, \pi^{-n})$, is defined as
\begin{align} \label{eq:costfcni}
	\nonumber&J_n^N(\pi^n, \pi^{-n}) :=\\
	&\hspace{-0.8em}\limsup_{T \rightarrow \infty} \frac{1}{T} \sum_{t=0}^{T-1} 
	\EE\bigg[  \Big\lVert Z_t^n-\frac{1}{N-1}\sum_{n' \neq n}\!Z_t^{n'} 
	\Big\rVert^2_{C_Z}\hspace{-1em}+ \big\lVert U_t^n \big\rVert^2_{C_U} \bigg],
\end{align}
where the norms for the state and control terms are taken with respect to the symmetric matrices $0 \leq C_Z\in\RR^{m\times m}, 0 < C_U\in\RR^{p\times p}$, respectively. The pair $(A,C_Z^{1/2})$ is assumed to be observable. The expectation in \eqref{eq:costfcni} is taken with respect to the probability measure induced by the joint control law $\pi$, the initial state distribution, and the noise statistics. 
The state average term in \eqref{eq:costfcni}, can be considered as a reference signal that each agent $n$ aims to track. We refer to this problem as a Linear Quadratic Tracking (LQT) 
problem.


The mean-field approach centers around the introduction of a generic (representative) agent that reacts to the average state, or mean-field trajectory, of the other agents. With some abuse of notation, the state of the generic agent at time $t$ is denoted by $Z_t$ which evolves as a function of control actions, denoted by $U_t\in\RR^p$, in an identical fashion to Eq. \eqref{eq:finitesystem}, i.e.,
\begin{align} \label{eq:dynamics}
Z_{t+1} = A Z_t +B U_t + W_t,
\end{align}
where 
$Z_0$ is generated by distribution $\Ns(\nu_0,\Sigma_0)$ and $W_t$ is an i.i.d. noise process generated according to the distribution $\Ns(0,\Sigma_w)
$, assumed to be independent of the agent's initial state. A generic agent's control policy at any time $t$, denoted by $\mu_t$, maps the generic agent's history at time $t$, given by $I_t = (Z_0,U_0, \ldots,U_{t-1},Z_t)$ 
to a control action $U_t = \mu_t(I_t )\in\RR^p$. The control policy $\mu_t$ is dependent, in an implicit (parametric) manner, on the \emph{mean-field trajectory} (i.e., average state trajectory of the other agents), given by $\bar{Z}=(\bar{Z}_0,\bar{Z}_1,\ldots)$. 
The collection of control policies across time is termed a control law and is denoted by $\mu =  (\mu_0,\mu_1,\ldots)\in\Ms$ where $\Ms$ is defined as the space of admissible control laws. 
The generic agent's expected cost under control law $\mu$, denoted by $J(\mu,
\bar{Z})$, is defined as
\begin{align} \label{eq:costfcn}
J(\mu, \bar{Z}) \!=\! \limsup_{T \rightarrow \infty} \frac{1}{T}\sum_{t=0}^{T-1} \EE \big[ c_t(Z_t, \bar{Z}_t, U_t) \big],
\end{align}
where $c_t(Z_t, \bar{Z}_t, U_t) = \lVert Z_t - \bar{Z}_t \rVert^2_{C_Z} + \lVert U_t \rVert^2_{C_U}$ is the instantaneous cost and the expectation is taken with respect to the control law $\mu$ and initial state and noise statistics. 
The mean-field trajectory $\bar{Z}$, is assumed to belong to the space of deterministic bounded sequences, that is, $\bar{Z}\in\Zs$ where $\Zs = \ell^\infty := \{ x = (x_0, x_1, \ldots) \mid \sup_{t\ge0}|x_t| < \infty \}$.\footnote{This assumption is validated in \cite{moon2014discrete}.} The mean-field trajectory in \eqref{eq:costfcn} can be viewed as a reference signal, resulting in an LQT problem.

To define an MFE, first define the operator $\Lambda:\Ms \to \Zs$ as a function 
from the space of admissible control laws $\Ms$ to the space of mean-field trajectories $\Zs$. Due to the information structure of the problem and the quadratic form of the cost function, 
the policy at any time depends only on the current state $Z_t$ 
and not all of the current information $I_t$ \cite{moon2014discrete}. 
Thus, $\Ms$ is the space of policies that maps the current state to a control action. Note that as \eqref{eq:dynamics}-\eqref{eq:costfcn} is an LQT problem, for a given $\bar{Z}$ the optimal control law will depend on $\bar{Z}$ in an open-loop manner. 
The operator $\Lambda$ is defined as follows: given $\mu \in \Ms$, the mean-field trajectory $\bar{Z} = \Lambda(\mu)$ is constructed recursively as
\begin{align}\label{eq:def_Lambda}
\bar{Z}_{t+1} = A \bar{Z}_{t} +B \mu_{t}(\bar{Z}_t),\quad \bar Z_0 = \nu_0.
\end{align}
If $\bar{Z} = \Lambda(\mu)$, then we refer to $\bar{Z}$ as the mean-field trajectory \emph{consistent with} $\mu$. Similarly, define an operator $\Phi: \Zs \to \Ms$ as a function from a mean-field trajectory to its optimal control law, also called the \emph{cost-minimizing 
controller},
\begin{align}\label{eq:cmc}
\Phi(\bar{Z}) := \argmin_{\mu\in\Ms} J(\mu, \bar{Z}).
\end{align}
The MFE can now be defined as follows.
\begin{definition} [\cite{saldi2018markov}] \label{def:mfe}
	The tuple $(\mu^*, \bar{Z}^*) \in \Ms \times \Zs$ is an MFE if $\mu^*= 
	\Phi(\bar{Z}^*)$ and $\bar Z^*=\Lambda(\mu^*)$. %
\end{definition} 
The trajectory $\bar{Z}^*$ is referred to as the \emph{equilibrium mean-field trajectory} and the controller $\mu^*$ as the \emph{equilibrium controller}. 
Note our MFE is non-stationary, in contrast to \cite{subramanian2019reinforcement,guo2019learning,fu2019actor}.\footnote{We allow for time-varying equilibrium mean-field trajectories. Refer to Definition (A3) in Section 2.2 of \cite{subramanian2019reinforcement} for clarification.} We refer to the corresponding game as a non-stationary LQ-MFG. 
By \cite{moon2014discrete}, the cost-minimizing controller in \eqref{eq:cmc} for any $\bar{Z} \in \Zs$ is given by $\mu' = (\mu'_1(Z_1;\bar Z),\mu'_2(Z_2;\bar Z),\ldots)= \Phi(\bar Z)$ 
with 
\begin{align} \label{eq:u_t}
\mu'_t(Z_t;\bar Z) = G_PPAZ_t + G_P\lambda_{t+1}(\bar Z),
\end{align}
where $G_P := - (C_U + B^T P B)^{-1} B^T$, $P$ is the unique positive definite solution to the discrete-time algebraic Riccati equation (DARE), 
\begin{align} \label{eq:s}
P =  A^T P A + C_Z +  A^T P B G_P P A 
\end{align}
and is guaranteed to exist \cite{bertsekas1995dynamic}. The sequence $\lambda \in \ell^{\infty}$ is generated according to,
\begin{align} \label{eq:costate}
\lambda_t(\bar Z) = -\sum_{k=0}^{\infty}  H_P^{k} C_Z \bar{Z}_{t+k},\,\text{ for }t=0,1,\ldots,
\end{align}
where $H_P := A^T (I + P B G_P)$. Substituting the cost-minimizing control, \eqref{eq:u_t}-\eqref{eq:costate}
, into the state equation of the generic agent, \eqref{eq:dynamics}, the closed-loop dynamics are given by
\begin{align*}
Z_{t+1} 
= H_P^T Z_t - B G_P \sum_{s=0}^{\infty} H_P^s C_Z \bar{Z}_{t+s+1} + W_t.
\end{align*}
By aggregating these dynamics over all agents and invoking Definition \ref{def:mfe}, the equilibrium mean-field trajectory obeys the following recursive expression,
\begin{align} \label{eq:eq_mf}
\bar{Z}^*_{t+1} = H_P^T \bar{Z}^*_t - B G_P \sum_{s=0}^{\infty} H_P^s C_Z 
\bar{Z}^*_{t+s+1},
\end{align}
for $t=0,1,\ldots$, where $\bar{Z}^*_0 = \nu_0$.

Under some mild conditions, the recursion in \eqref{eq:eq_mf} exhibits desirable properties that allow conversion of the LQT problem of \eqref{eq:costfcn} to be expressed as an LQR 
problem (described in the following section). To illustrate these properties, let $M$ be a square matrix of dimension $m$ and define the operator $\Ts:\RR^{m\times m}\to\RR^{m\times m}$ as
\begin{align}\label{eq:Foperator}
\Ts(M) := H_P^T-BG_P\sum_{s=0}^\infty H_P^sC_ZM^{s+1}.
\end{align}
Consider a matrix $F^* \in \RR^{m \times m}$ s.t. $F^* = \Ts(F^*)$; then a candidate for $\bar{Z}^*$ can be characterized by $F^*$ as its mean-field state matrix i.e. $\bar{Z}^*_{t+1} = F^* \bar{Z}^*_t$. We prove that under the following assumption, $F^*$ uniquely determines $\bar{Z}^*$.
\begin{assumption}   \label{asm:boundforcontract}
	Given $A,B,C_Z,C_U$ and $G_P$, $H_P$, where $P$, is the unique 
	positive definite solution of \eqref{eq:s}, we have
	$$T_P:=  \lVert H_P \rVert_2 + 
	\frac{\lVert B G_P \rVert_2 \lVert C_Z \rVert_2}{(1 - \lVert H_P 
	\rVert_2)^2} < 1.$$
\end{assumption}
\noindent Assumption \ref{asm:boundforcontract} above is motivated from the literature \cite{uz2020approximate}, \cite{moon2014discrete}. It is stronger than the standard assumptions, e.g., \cite{uz2020approximate}, but gives rise to desirable linearity properties of the MFE as shown in Proposition \ref{lem:exist_unique_linear} below. This enables the conversion of the LQT problem \eqref{eq:dynamics}-\eqref{eq:costfcn} 
into a state-feedback LQG problem. This conversion is core to the construction of our RL algorithm. Also, since Assumption 1 below implies the primary assumption in \cite{uz2020approximate}, the existence and uniqueness of the MFE is ensured. 
\begin{proposition} \label{lem:exist_unique_linear}
Under Assumption \ref{asm:boundforcontract}, there exists a unique equilibrium mean-field trajectory $\bar{Z}^*$. Furthermore, $\bar{Z}^*$ follows linear dynamics, that is, there exists an $F^* \in \mathbb{F} := \{F \in \RR^{m \times m} : \lVert F \rVert_2 \leq (1+T_P)/2\}$, such that $\bar{Z}^*_{t+1} = F^* \bar{Z}^*_{t}$ for $t = 0,1,\ldots$, with $\bar{Z}_0^* = \nu_0$.
\end{proposition}
\begin{proof}
As Assumption \ref{asm:boundforcontract} above implies Assumption 1 in \cite{uz2020approximate}, the proof of existence and uniqueness of the MFE is obtained in a similar manner.
%
%
To prove that the equilibrium mean-field trajectory evolves linearly, the operator $\Ts$ is shown to be contractive on $\mathbb{F}$. Let $F_1, F_2 \in \mathbb{F}$,
\begin{align} \label{eq:prop_1_work}
& \lVert \Ts(F_1) - \Ts(F_2) \rVert_2 \nonumber \\
& \leq \lVert B G_P \rVert_2 \lVert C_Z \rVert_2 \sum_{s=0}^{\infty} \lVert H_P \rVert^s_2 \big\lVert F_1^{s+1} - F_2^{s+1} \big\rVert_2 ,\nonumber \\
& = \frac{\lVert B G_P \rVert_2 \lVert C_Z \rVert_2}{\lVert H_P \rVert_2} \sum_{s=1}^{\infty} \lVert H_P \rVert^s_2 \big\lVert F_1^{s} - F_2^{s} \big\rVert_2 ,\nonumber \\
& \leq \frac{\lVert B G_P \rVert_2 \lVert C_Z \rVert_2}{\lVert H_P \rVert_2} \sum_{s=1}^{\infty} \lVert H_P \rVert^s_2 \big\lVert F_1 - F_2 \big\rVert_2 \nonumber \\
& \hspace{5.cm} \sum_{r=0}^{s-1} \lVert F_1 \rVert^{s-r-1}_2 \lVert F_2 
\rVert^r_2 ,\nonumber \\
& < \frac{\lVert B G_P \rVert_2 \lVert C_Z \rVert_2}{\lVert H_P \rVert_2} 
\sum_{s=1}^{\infty} s\lVert H_P \rVert^s_2 \big\lVert F_1 - F_2 \big\rVert_2 ,\nonumber \\
& = \frac{\lVert B G_P \rVert_2 \lVert C_Z \rVert_2}{(1 - \lVert H_P 
\rVert_2)^2} \big\lVert F_1 - F_2 \big\rVert_2 .
\end{align}
The second inequality is obtained by the fact that $\lVert F_1 \rVert_2, \lVert F_2 \rVert_2 < 1$ and that for any two square matrices $A,B$, and any $k \in \mathcal{Z}$, $A^{k} - B^{k} = \sum_{l=0}^{k-1} A^{k-l-1} (A - B) B^{l}$. Hence under Assumption \ref{asm:boundforcontract}, the operator $\Ts$ is a contraction mapping. As $\mathbb{F}$ is a complete metric space, using the Banach fixed point theorem, we can deduce the existence of $F^* \in \mathbb{F}$ s.t. $F^* = \mathcal{F}(F^*)$. Hence if we define a sequence $\bar{Z}^*$ s.t. $\bar{Z}^*_0 = \mu_0$ and $\bar{Z}^*_{t+1} = F^* \bar{Z}^*_t$, then it satisfies the dynamics of the equilibrium mean-field trajectory \eqref{eq:eq_mf} and as the equilibrium mean-field trajectory is unique, it follows linear dynamics.
\end{proof}
Notice that as $F^* \in \FF$, $\bar{Z}$ is asymptotically stable. The following property of $\Ts$ will be useful later.
\begin{lemma} \label{lem:F_contract}
Under Assumption \ref{asm:boundforcontract}, if $F \in \FF$ then $\lVert \Ts(F) \rVert_2 \leq T_P$, and hence $\Ts(F) \in \FF$, for all $F \in \FF$.
\end{lemma}
\emph{Proof.} Let $F \in \FF$; then,
\begin{align*}
\lVert \Ts(F) \rVert_2 & \leq \lVert H_P \rVert_2 + \lVert B G_P \rVert_2 \lVert 
C_Z \rVert_2 \sum_{s=0}^{\infty} \lVert H_P \rVert_2^s \lVert F \rVert_2^{s+1}  \\
& = \lVert H_P \rVert_2 + \frac{\lVert B G_P \rVert_2 \lVert C_Z \rVert_2 \lVert F \rVert_2}{1 - \lVert H_P \rVert_2 \lVert F \rVert_2} \\
& \leq \lVert H_P \rVert_2 + \frac{\lVert B G_P \rVert_2 \lVert C_Z \rVert_2}{1 - \lVert H_P \rVert_2 } < T_P \hspace{5em}\hfill\square
\end{align*}
While agents are aware of the functional forms of the dynamics and cost functions, no agent has knowledge of the true model parameters. As such, we aim to develop an RL algorithm for learning the MFE in the absence of model knowledge. The remainder of the paper is devoted to this task.

\section{Actor-Critic Algorithm for Non-stationary LQ-MFGs}
\label{sec:RL}



The fact that the equilibrium mean-field trajectory follows linear dynamics enables us to develop  RL algorithms for solving the \emph{non-stationary} MFG, 
in contrast to the stationary MFG in \cite{fu2019actor} 
Specifically, as a result of Proposition \ref{lem:exist_unique_linear}, it suffices to find  the MFE  by searching over the set of stable matrices $\FF$ defined therein. Moreover, given any mean-field trajectory $\bar Z$ parameterized by its mean-field state matrix $F \in \FF$, the LQT problem in \eqref{eq:costfcn} can be written as a state-feedback LQG problem with an augmented state $X_t = (Z^T_t,\bar{Z}^T_t)^T$. The augmented state follows linear dynamics
\begin{equation} \label{eq:augdynamics}
X_{t+1} = \bar{A} X_t + \bar{B} U_t + \bar W_t, \hspace{0.2cm}
\end{equation}
with
\begin{equation}\label{eq:augdynamics_2}
\bar{A} = \left( \begin{array}{cc}
A & 0 \\ 0 & F
\end{array} \right), \hspace{0.2cm} \bar{B} =  \left( \begin{array}{c}
B \\ 0
\end{array} \right), \hspace{0.2cm} \bar W_t = \left( \begin{array}{c}
W_t \\ 0
\end{array} \right),
\end{equation}
where $W_t$ is the noise term  in \eqref{eq:dynamics} and consequently $\bar{W}_t \sim \Ns(0,\Sigma_{\bar{w}})$, and $\Sigma_{\bar{w}} = [I \hspace{0.2cm} 0]^T \Sigma_{w} [I \hspace{0.2cm} 0]$. Recall that $F \in \FF$ and as a result it is stable. Accordingly, the cost in \eqref{eq:costfcn} can be written as 
\begin{align} \label{eq:augcostfunc} 
J(K,F) = \limsup_{T \rightarrow \infty}\frac{1}{T}\sum_{t=0}^{T-1} \EE[ c_t(X_t,U_t) ], 
\end{align}
where $c_t(X_t,U_t) = \lVert X_t \rVert^2_{C_X}  + \lVert U_t \rVert^2_{C_U}$, $U_t = - K X_t + \zeta_t$ where the structure is motivated by \eqref{eq:u_t} (with an added exploration term), and 
\begin{align*}
C_X=\left[ \begin{array}{cc}C_Z & -C_Z \\ -C_Z & C_Z\end{array} \right],
\end{align*}
is positive semi-definite. Notice that in the RL setting the controller is closed-loop with the mean-field trajectory \cite{subramanian2019reinforcement,guo2019learning}, in contrast to the full knowledge setting where it has implicit dependence on the mean-field trajectory. 
With some abuse of notation (in relation to \eqref{eq:costfcn}), the cost functional in \eqref{eq:augcostfunc} takes as input the matrix $K$, replacing the control law $\mu$, and the matrix $F$, replacing the mean-field trajectory $\bar Z$ (as a result of Proposition \ref{lem:exist_unique_linear}). As an upshot of the reformulation, the cost-minimizing controller given $F$ can be obtained in a model-free way using RL algorithms that solve state-feedback LQG problems. Hence, by the NCE principle \cite{huang2007large}, 
the MFE can be approximated in a model free setting, by recursively: 
(1) finding the \emph{approximate} cost-minimizing controller $K$ for the system \eqref{eq:augdynamics}-\eqref{eq:augcostfunc}
, and (2) updating mean-field state matrix $F$ in \eqref{eq:augdynamics_2}.

We first deal with finding the cost-minimizing controller for the system \eqref{eq:augdynamics}-\eqref{eq:augcostfunc} in a model-free setting. One method to achieve that is RL for 
state-feedback LQG problem. 
The recent work of \cite{yang2019global} uses a natural policy gradient actor-critic method to solve such a problem, albeit the MC (for sampling) is assumed to be \emph{fully mixed}. 
We adapt this method for an \emph{unmixed but fast-mixing} MC setting 
to find the approximate cost-minimizing controller.

We briefly outline the actor-critic algorithm \cite{yang2019global} and our modification to account for the unmixed MC. 
Each iteration $s \in [S]$ of the algorithm involves two steps, namely the \emph{actor} and the \emph{critic}. The critic observes, the state of the system $X_t$, the control actions $U_t = -K X_t + \zeta_t$ (where $\zeta_t$ is an i.i.d. Gaussian noise for exploration and $K$ is a stabilizing controller
), and the instantaneous cost $c_t$ for $t \in \{0,\ldots,T-1\}$.


The fundamental modification is that by having the total number of timesteps $T$ also conditional on the initial state $X_0$, we can prove convergence of the critic step for the unmixed but fast-mixing 
MC setting. 
This dependence 
is presented in Proposition \ref{prop:PE} in the next section. 
The critic produces an estimate $\hat{\theta}$ of the parameter vector $\theta$ which characterizes the action-value function 
pertaining to $K$. Once the estimate $\hat\theta$ is obtained, the \emph{actor} updates the controller in the direction of the natural policy gradient as given in \cite{yang2019global}. 
After $S$ actor-critic updates we arrive at the approximate cost-minimizing controller for system \eqref{eq:augdynamics}-\eqref{eq:augcostfunc}. As per \cite{yang2019global}, the approximate cost-minimizing controller is close to the actual cost-minimizing controller, provided that $S$ and $T$ are large enough.

Now we 
describe the update of the mean-field state matrix $F$ in \eqref{eq:augdynamics_2} given the cost-minimizing controller $K$ (computed in the previous step). The \emph{state aggregator} is a simulator which computes the new mean-field state matrix $F'$, given $K$, by simulating the mean-field trajectory consistent with controller $K$. Hence, it fulfills the role of operator $\Lambda$ for linear feedback controllers. The state aggregator is similar to the simulators used in \cite{guo2019learning,elie2019approximate}. To obtain $F'$, we first model the behavior of a generic agent with dynamics \eqref{eq:dynamics}, under controller $K$,
\begin{align}\label{eq:state_agg_0}
Z_{t+1} & = A Z_t + B (-K X_t) + W_t,  \nonumber \\
 & = (A - B K_1) Z_t - B K_2 \bar{Z}'_t + W_t, 
\end{align} 
where  $X_t = [Z_t^T, (\bar{Z}'_t)^T]^T$ and $K = [K_1^T, K_2^T]^T$. Notice that the controller $K$ is online 
with respect to the mean-field trajectory $\bar{Z}'$ as per the definition of $K$. By aggregating \eqref{eq:state_agg_0}, the updated mean-field trajectory $\bar Z'$ is shown to follow linear dynamics: 
\begin{align} \label{eq:state_agg}
\bar{Z}'_{t+1} = F' \bar{Z}'_t, \text{ where } F' = A - B(K_1 + K_2).
\end{align}
Hence the state aggregator updates the mean-field state matrix to $F'$ in equation \eqref{eq:augdynamics_2} given the cost-minimizing controller $K$. In the next section we show that if $F$ is stable, $F'$ will be stable as well.

The combination of the actor-critic algorithm for state-feedback LQG and the state aggregator \eqref{eq:state_agg}, as outlined in Algorithm \ref{alg:actorcritic}, 
performs an approximate and data-driven update of the operator $\Ts$ (as in \eqref{eq:eq_mf}). In Section \ref{sec:Analysis}, we prove finite-sample bounds to show convergence of Algorithm \ref{alg:actorcritic}. The critic and actor steps are standard, and thus details have been omitted; see \cite{konda2000actor,fu2019actor}.
\begin{algorithm}
	\caption{Actor-critic for LQ-MFG} \label{alg:actorcritic}
	\begin{algorithmic}[1]
		\STATE {\bf Input}: Number of iterations: $R$, $\{S_r : r \in [R]\}$, $\{T_{s,r}: s \in [S_r], r \in [R] \}$.
		\STATE {\bf Initialize}: $F^{(1)} \in \mathbb{F}$ and stabilizing 
		$K^{(1,1)}$
		\FOR {$r \in [R]$}
		\FOR {$s \in [S_r]$}
		
		\STATE {\bf Critic Step} Compute $\hat\theta^{(s,r)}$ using $X_t, U_t$ and $c_t$ for $t \in \{0,\ldots,T_{(s,r)} - 1\}$ 
		
		\STATE {\bf Actor Step} Compute $K^{(s+1,r)}$ using $\hat\theta^{(s,r)}$ and $K^{(s,r)}$ 
		
		\ENDFOR
		
		\STATE {\bf Update mean-field trajectory} $F^{(r+1)}$ using state aggregator and $K^{(S_r,r)}$ \label{state:aggregate} (by \eqref{eq:state_agg})
		
		\STATE $K^{(1,r+1)} \leftarrow K^{(S_r,r)}$
		\ENDFOR
		\STATE {\bf Output}: $K^{(S_R,R)}, F^{(R)}$
	\end{algorithmic}
\end{algorithm}

\section{Analysis}
\label{sec:Analysis}

We now provide non-asymptotic convergence guarantees for Algorithm \ref{alg:actorcritic}. Moreover, we also provide an error bound for the approximate MFE output, as generated by Algorithm 1, with respect to the NE of the finite population game. 

\subsection{Non-asymptotic convergence}
We begin by presenting the convergence result of the critic step in the algorithm. The output of this step is the parameter vector estimate $\hat\theta^{(s,r)}$ which is shown to be arbitrarily close to the true parameter vector $\theta^{(s,r)}$ given that the number of time-steps in the critic step $T_{s,r}$ is sufficiently large. 
\begin{proposition} \label{prop:PE}
	For any $r \in [R]$ and $s \in [S_r]$, the parameter vector estimate $\hat{\theta}^{(s,r)}$ satisfies
	\begin{align*}
	\lVert \hat\theta^{(s,r)} - \theta^{(s,r)} \rVert^2_F \leq \kappa^{(s,r)}_1   \frac{ \log^6 T_{s,r} }{\sqrt{T_{s,r}}},
	\end{align*}
	with probability at least $1-T^{-4}_{s,r}$. The variable $\kappa^{(s,r)}_1$ 
	is a polynomial in  the initial state $X_0$ and controllers $K^{(1,r)}$ and $K^{(s,r)}$.
\end{proposition}
\begin{proof}
	The proof is an adaptation of the technique used in \cite{yang2019global}. 
	We provide the main idea of the proof and how it is modified to  the unmixed 
	but fast-mixing Markov chain setting. The problem of estimating the parameter vector $\theta^{(s,r)}$ is first formulated as a minimax optimization problem. Then, the estimation error $\lVert \hat\theta^{(s,r)} - \theta^{(s,r)} \rVert^2_F$, is shown to be upper bounded by the duality gap of that minimax problem. Using results from \cite{wang2017finite} an explicit expression is obtained for the duality gap.
	
	Most of the proof from Theorem 4.2 \cite{yang2019global} follows except for two differences: 1) A lower bound on the probability of the event $\Es$ 	in 	equation (5.20) needs to be established, and 2) An upper bound for the 	expression $\EE [\lVert \phi_t \phi^T_{t+1} \hat\theta \rVert^2_2]$, where $\phi_t := \svec \Big( \big(X_t^T, U_t^T \big)^T \big( X_t^T, U_t^T \big) \Big)$, 
	needs to be found. First we lower bound the probability of the event $\Es_{t}$. In the following, we define $T =  T_{s,r}$ for clarity. Let us define the event $\Es$,
	\begin{align*}
	\Es := \bigcap_{t=0}^{T-1}\big\{ \big\lvert \big\lVert X_{t} \big\rVert_2^2 + \big\lVert 	
	U_{t} \big\rVert_2^2 - \EE\big[ \big\lVert X_{t} \big\rVert_2^2 + \big\lVert 
	U_{t} 	\big\rVert_2^2 \big] \big\rvert < C_1 \big\}. 
	\end{align*} 
	Let us first define an event $\Es_{t}$ such that,
	\begin{align*}
	& \Es_{t} := \big\{ \big\lvert \big\lVert X_{t} \big\rVert_2^2 + \big\lVert U_{t} \big\rVert_2^2 - \EE\big[ \big\lVert X_{t} \big\rVert_2^2 + \big\lVert U_{t} \big\rVert_2^2 \big] \big\rvert \geq C_{1,t} \big\},
	\end{align*}
	for some $C_{1,t}>0$ and $t \in \{0,\ldots,T-1\}$. We want to establish an upper bound for the probability $\PP(\Es_{t})$. Suppose $K$ is a stabilizing controller for \eqref{eq:augdynamics} and define $L_K := \bar{A} - \bar{B}K$. Now let us consider control policies of the form 
	\begin{align*}
	U_{t} = -K X_{t} + \zeta_{t},
	\end{align*}
	where $\zeta_{t}$ is generated i.i.d. with distribution $\Ns(0,\sigma^2 I)$. The marginal distributions of $X_{t}$ and $U_{t}$ under filtration $\Fs_{0}$ (having observed $X_0$) is,
	\begin{align} \label{eq:X_t}
	X_{t} & \sim \Ns\Big(L_K^{t} X_0, \Sigma_{X,t} \Big), t \in \{0,\ldots,T-1\},\\
	U_{t} & \sim \Ns\Big(K L_K^{t} X_0, K \Sigma_{X,t} K^T + \sigma^2 I 
	\Big), \nonumber
	\end{align}
	where 
	\begin{align*}
	\Sigma_{X,0} = 0, \Sigma_{X,t} = \sum_{s=0}^{t-1} L_K^s \Sigma_X (L_K^s)^T, t \in \{1,\ldots,T-1\}
	\end{align*} 
	and $\Sigma_X = \Sigma_{\bar{w}} + \sigma^2 \bar{B} \bar{B}^T$. Let 
	\begin{align*}
	\tilde{X}_{t} := X_{t} - L_K^{t} X_0, \tilde{U}_{t}:= U_{t} - K L_K^{t} X_0.
	\end{align*}
	By definition $\tilde{X}_{t}$ and $\tilde{U}_{t}$ are zero mean Gaussian random vectors. Now we consider the quantity,
	\begin{align} \label{eq:expr_1}
	& \big\lVert X_{t} \big\rVert_2^2 + \big\lVert U_{t} 	\big\rVert_2^2 - \EE\big[ \big\lVert X_{t} \big\rVert_2^2 + \big\lVert U_{t} 	\big\rVert_2^2 	\big] = \big\lVert \tilde{X}_{t} \big\rVert_2^2 + \big\lVert \tilde{U}_{t} 	\big\rVert_2^2 - \nonumber \\
	& \EE\big[ \big\lVert \tilde{X}_{t} \big\rVert_2^2 + \big\lVert \tilde{U}_{t} \big\rVert_2^2 \big]  +  2 	\big\langle L^{t-1} X_t, \tilde{X}_{t} \big\rangle + 2 	\big\langle K L^{t-1} X_t, \tilde{U}_{t} \big\rangle. 
	\end{align}
	Let $V_{t} = [\tilde{X}^T_{t}, \tilde{U}^T_{t}]^T$ which, using \eqref{eq:X_t} and the definition of $V_{t}$, has the distribution $\Ns(0,\Sigma^{t}_v)$, where
	\begin{align*}
	\Sigma^{t}_v & = \Bigg(\begin{array}{cc}
	\Sigma_{X,t} & -\Sigma_{X,t} K^T \\
	-K \Sigma_{X,t} &  K \Sigma_{X,t} K^T + \sigma^2 I
	\end{array} \Bigg).
	\end{align*}
	Using the definition of $\Sigma_{X,t}$ we deduce
	\begin{align*}
	\Sigma^{t+1}_v -  \Sigma^{t}_v = \Bigg(\begin{array}{c}
	I \\
	-K 
	\end{array} \Bigg) L^{t}_K \Sigma_X (L^{t}_K)^T
	\Bigg(\begin{array}{c}
	I \\
	-K 
	\end{array} \Bigg)^T ,
	\end{align*}
	which means $\Sigma^{t+1}_v \geq \Sigma^{t}_v$. Moreover, for a stabilizing $K$, $X_{t}$ and consequently $V_t$ converges to a stationary distribution as $t \rightarrow \infty$. Let us denote the covariance of this distribution by $\Sigma_v^{\infty}$. And since $\Sigma^{t+1}_v \geq \Sigma^{t}_v \implies \Sigma^{\infty}_v \geq \Sigma^{t}_v$, using $V_{t}$ the expression in \eqref{eq:expr_1} can be written as,
	\begin{align*}
	& \lVert X_{t} \rVert_2^2 + \lVert U_{t} \rVert_2^2 - \EE\big[ \lVert X_{t} \rVert_2^2 + \lVert U_{t} \rVert_2^2 \big] \\
	& = \lVert V_{t} \rVert_2^2 - \EE \big[ \lVert V_{t} \rVert_2^2 \big] + \underbrace{2 \Big\langle \Big[\begin{array}{c} L^{t} X_0 \\ K L^{t} X_0 \end{array}\Big], V_{t} \Big\rangle}_{v_{t}}.
	\end{align*}
	By definition $v_{t}$ has the distribution $\Ns (0, \sigma_v^t )$ where
	\begin{align*}
	\sigma_v^t = 4\bigg[\begin{array}{c} L^{t} X_0 \\ K L^{t} X_0 \end{array}\bigg]^T \Sigma^{t}_v \bigg[\begin{array}{c} L^{t} X_0 \\ K 	L^{t} X_0 \end{array}\bigg].
	\end{align*}
	This is upper bounded by,
	\begin{align}
	\sigma_v^t & \leq 4(1+\lVert K \rVert^2_F) \lVert\Sigma^{t}_v \rVert_2 \lVert X_0 \rVert^2_2 \nonumber\\
	& \leq 4(1+\lVert K \rVert^2_F) \lVert\Sigma^{\infty}_v \rVert_2 \lVert X_0 
	\rVert^2_2 := \sigma^{\infty}_v.
	\end{align}
	Now let us define events $\Es^1_{t}$, $\Es^2_{t}$, 
	$\Es^3_{t}$ 
	and 
	$\Es^4_{t}$ for $t \in \{0,\ldots,T-1\}$ 
	such that
	\begin{align*}
	& \Es^1_{t} := \big\{ \big\lvert \lVert V_{t} \rVert_2^2 - \EE \big[ \lVert V_{t} \rVert_2^2 \big] \big\rvert +  \big\lvert v_{t} \big\rvert \geq C_{1,t} \big\}, \\
	& \Es^2_{t} := \big\{ \big\lvert  \lVert V_{t} \rVert_2^2 - \EE \big[ \lVert V_{t} \rVert_2^2 \big] \big\lvert < C_{2,t} \,\,\, \& \,\,\, \big\lvert v_{t} \big\rvert < C_3 \big\},  \\
	& \Es^3_{t} := \big\{ \big\lvert  \lVert V_{t} \rVert_2^2 - \EE \big[ \lVert V_{t} \rVert_2^2 \big] \big\lvert \geq C_{2,t} \big\}, \hspace{0.25cm}\Es^4_{t} := \big\{ \big\lvert v_{t} \big\rvert \geq C_3 \big\},
	\end{align*}
	where $C_{2,t} , C_3 > 0$ and $C_{1,t} = C_{2,t} + C_3$. From triangle 
	inequality, we have
	\begin{align*}
	\big\lvert \lVert V_{t} \rVert_2^2 - \EE \big[ \lVert V_{t} 	\rVert_2^2 	\big] + v_{t} \big\rvert \leq \big\lvert \lVert V_{t} \rVert_2^2 - \EE \big[ \lVert 	V_{t} 	\rVert_2^2 \big] \big\rvert + \big\lvert v_{t} \big\rvert
	\end{align*}
	and therefore $\Es_{t} \subset \Es^1_{t}$. Moreover, $\Es^1_{t} \cap \Es^2_{t} = 	\varnothing \implies \Es^1_{t} \subset \big(\Es^2_{t} \big)^c \implies 	\Es_{t} \subset \big(\Es^2_{t} \big)^c = \Es^3_{t} \cup \Es^4_{t}$, and hence
	\begin{align} \label{eq:upper_bound_events}
	\PP(\Es_{t}) \leq \PP(\Es^3_{t} \cup \Es^4_{t}) \leq \PP(\Es^3_{t}) + 
	\PP(\Es^4_{t}).
	\end{align} 
	First we establish an upper bound for $\PP (\Es^3_{t})$ for $t \in \{0,\ldots,T-1\}$. 
	Using Hanson-Wright inequality \cite{rudelson2013hanson},
		\begin{align} \label{eq:P_eps_3}
		\PP(\Es^3_{t}) \leq 2   e^{-C   \min \big((C_{2,t})^2   \lVert 
			\Sigma^{t}_v \rVert_F^{-2}, C_{2,t}   \lVert \Sigma^{t}_v 
			\rVert^{-1}_2 \big)},
		\end{align}
		for a $C>0$. Substituting 
		\begin{align*}
		C_{2,t} = k_1 \log(T/3) \lVert \Sigma^{t}_v \rVert_2, \hspace{0.2cm} k_1 \geq \max\bigg(\frac{6}{C}, \frac{p+2m}{\log (T/3)}\bigg)
		\end{align*}
		and using a method similar to \cite{yang2019global} we get
		\begin{align} \label{eq:P_eps_3_1}
		& (C_{2,t})^2   \lVert \Sigma^{t}_v \rVert_F^{-2} = k^2_1   (\log (T/3))^2 \lVert \Sigma^{t}_v \rVert_2^2   \lVert \Sigma^{t}_v \rVert_F^{-2} \geq \\
		&  k^2_1   (\log (T/3))^2   (p+2m)^{-1} \geq k_1   \log (T/3) =  C_{2,t} \lVert \Sigma^{t}_v \rVert^{-1}_2. \nonumber
		\end{align}
		The first inequality is due to the relation between the induced and 
		Frobenius norms and the second inequality is due to definition of $k_1$. Using equations \eqref{eq:P_eps_3}, \eqref{eq:P_eps_3_1} and the expression for $C_{2,t}$, 
		\begin{align} \label{eq:P_eps_1}
		\PP(\Es^3_{t}) \leq \frac{2 T^{-C   k_1}}{3} \leq \frac{2 T^{-6}}{3},
		\end{align}
	where $C$ is an absolute constant. Now we establish an upper bound for $\PP(\Es^4_{t})$. Using Gaussian tail bounds, 
	\begin{align*}
	\PP(\Es^4_{t}) \leq 2 e^{ -C^2_3/(2\sigma_v^t) }.
	\end{align*}
	Substituting 
	\begin{align*}
	C_3 = k_2   \log (T/6)   \sqrt{2\sigma^{\infty}_v}, \hspace{0.2cm}k_2 \geq \max (6,(\log (T/6))^{-1}),
	\end{align*}
	we get
	\begin{align}  \label{eq:P_eps_2}
	\PP(\Es^4_{t}) \leq \frac{T^{-k_2}}{3} \leq \frac{T^{-6}}{3}.
	\end{align}
	The first inequality is due to $\frac{C^2_3}{2\sigma_v^t} \geq k_2^2 (\log (T/6))^2 \geq k_2 \log (T/6)$. Now using \eqref{eq:upper_bound_events}, 
	\eqref{eq:P_eps_1} 
	and \eqref{eq:P_eps_2}, we deduce
	\begin{align}
	\PP(\Es_{t}) \leq T^{-6}.
	\end{align}
	Finally we obtain a time invariant upper bound on $C_{1,t}$:
	\begin{align}
	C_{1,t} & = C_{2,t} + C_3  \\
	& \leq k_1 \log(T/3) \lVert \Sigma^{\infty}_v \rVert_2 + k_2 \log (T/6) \sqrt{2\sigma^{\infty}_v} =: C_1 . \nonumber
	\end{align}
	Let us define another event,
	\begin{align*}
	\Es^*_{t} := \big\{ \big\lvert \big\lVert X_{t} \big\rVert_2^2 
	+ 
	\big\lVert U_{t} \big\rVert_2^2- \EE\big[ \big\lVert X_{t} \big\rVert_2^2 + 
	\big\lVert 
	U_{t} 
	\big\rVert_2^2 \big] \big\rvert \geq C_1 \big\}.
	\end{align*}
	Since $C_1 \geq C_{1,t}$, $\Es^*_{t} \subset \Es_{t} \implies 
	\PP(\Es^*_{t}) \leq P(\Es_{t}) \leq T^{-6}$. Moreover, $\PP(\cup_{i=0}^{T-1} \Es^*_{t} ) \leq \sum_{i=0}^{T-1} \PP(\Es^*_{t}) \leq T^{-5}$. 
	Now we can bound 
	\begin{align*} 
	\PP(\Es) = 1 - \PP(\cup_{i=0}^{T-1} \Es^*_{t}) \geq 1 - T^{-5} .
	\end{align*}
	Now we bound the expression $\EE [\lVert \phi_t \phi^T_{t+1} 
	\hat\theta \rVert^2_2]$ 
	in proof of Theorem 4.2 
	in \cite{yang2019global}:
	\begin{align} \label{eq:bound_quant}
	& \EE [\lVert \phi_t \phi^T_{t+1} \hat\theta \rVert^2_2] = \EE [\hat\theta^T 
	\phi_{t+1} 	\phi^T_t \phi_t \phi^T_{t+1} \hat\theta]  \nonumber \\ 
	& = \EE [\phi^T_t \phi_t \hat\theta^T \phi_{t+1} \phi^T_{t+1} 
	\hat\theta] \leq \EE \big[\phi^T_t \phi_t  \big] \lVert \hat\theta 
	\rVert^2_2 \EE 
	\big[ \lVert \phi_{t+1} \phi^T_{t+1} \rVert_2 \big] \nonumber \\
	& \leq \EE \big[ \phi^T_t \phi_t  \big] 
	\lVert \hat\theta \rVert^2_2 \EE \big[ \lVert \phi_{t+1} \phi^T_{t+1} 
	\rVert_F 
	\big] \nonumber\\
	& \leq \EE \big[ \phi^T_t \phi_t  \big] \lVert \hat\theta \rVert^2_2 \EE 
	\big[ 
	\lVert \phi_{t+1} \rVert^2_F \big] \nonumber\\
	& = \EE \big[ \phi^T_t \phi_t  \big] \lVert \hat\theta \rVert^2_2 \EE \big[ 
	\tr(\phi_{t+1} \phi^T_{t+1} )\big] \nonumber\\ 
	& = \EE \big[ \phi^T_t \phi_t  \big] \lVert \hat\theta \rVert^2_2 \EE \big[ 
	\phi^T_{t+1} \phi_{t+1}\big].
	\end{align}
	The first inequality is due to the Cauchy-Schwarz inequality. Using the 
	definition of $\phi_t$,
	$\phi^T_t \phi_t = 
	\big\lVert \big( X_t^T, U_t^T \big)^T \big\rVert^4_2$. 
	Hence an upper bound for the expectations on the RHS of \eqref{eq:bound_quant} can be obtained by using the fourth moments of the marginal distributions of $X_t$ and $U_t$ given in \eqref{eq:X_t}. Hence the bound on $\EE [\lVert \phi_t \phi^T_{t+1} \hat\theta \rVert^2_2]$ would be a function of the initial state $\lVert X_1 \rVert_2$, the stationary distribution $\Sigma_v^{\infty}$ and $\sigma^{\infty}_v$, and the bound on $\lVert \hat\theta \rVert^2_2$, imposed by the critic \cite{yang2019global} using a projection operator.
\end{proof}

Note that in Proposition \ref{prop:PE}, 
$\kappa^{(s,r)}_1$ is a polynomial in 
the initial state, $X_0$. 
This dependence is due to the MC not being fully mixed. 
Having proved a finite\kzedit{-}sample bound on the estimation error for the critic step, 
we now state the convergence guarantee for the actor-critic algorithm for fixed mean-field trajectory. In particular, the approximate cost-minimizing controller found by the actor-critic $K^{(S_r,r)}$, can be brought arbitrarily close to the cost-minimizing controller $K^{(*,r)}$, by choosing the number of iterations of critic, $T_{s,r}$, and actor-critic, $S_r$, sufficiently large. 
\begin{proposition} \label{prop:RL-LQG}
	For any $r \in [R]$, let $K^{(1,r)}$ be a stabilizing controller and $S_r$ be chosen such that $S_r \geq  \kappa_2 \lVert \Sigma_{K^{(*,r)}} \rVert_2 \log\big( 2\big(J(K^{(1,r)},F^{(r)}) - J(K^{(*,r)},F^{(r)})\big) / \epsilon_r \big)$, 
	for any $\epsilon_r > 0$, where $\Sigma_{K^{(*,r)}}$ is the covariance matrix of stationary distribution induced by controller $K^{(*,r)}$. Moreover, let $T_{s,r} \geq (\kappa^{(s,r)}_1)^{5/2} \epsilon^{-5}_r$ for $s \in S_r$. Then, with probability at least $1 - \epsilon^{10}_r$,
	\begin{align*}
	& J(K^{(S_r,r)},F^{(r)}) - J(K^{(*,r)},F^{(r)}) \leq \epsilon_r, \\
	& \hspace{0.9cm} \lVert K^{(S_r,r)} - K^{(*,r)} \rVert_F \leq \sqrt{ \kappa_3 \epsilon_r }
	\end{align*}
	and $K^{(s,r)}$ are stabilizing for $s \in S_r$. The variable $\kappa^{(s,r)}_1$ 
	is dependent on $K^{(1,r)}, K^{(s,r)}$ and initial state $X_0$ (as in Proposition \ref{prop:PE}). The parameter $\kappa_2$ and $\kappa_3$ are absolute constants. 
\end{proposition}
The first inequality and the stability guarantee in Proposition \ref{prop:RL-LQG} follows from Theorem 4.3 in \cite{yang2019global} and the second inequality can be obtained from proof of Lemma D.4 in \cite{fu2019actor}. 
Next, we provide the non-asymptotic convergence guarantee 
for Algorithm \ref{alg:actorcritic}. We prove that the output of Algorithm \ref{alg:actorcritic}, also called the \emph{approximate} MFE $(K^{(S_R,R)}, F^{(R)})$, approaches the MFE of the LQ-MFG $(K^*, F^*)$. We also provide an upper bound on the difference in cost $J$ under the approximate and the  exact MFE.

\begin{theorem} \label{thm:conv_guarantee}
	For any $r \in [R]$, let $\epsilon_r$ be defined as
	\begin{align}
	& \epsilon_r = \kappa_3^{-1} \min \Bigg( \frac{\epsilon^2}{2^{2r+4} \lVert B \rVert_2^2}, \frac{\epsilon^2}{2}, \frac{(1-T_P)^2}{ 8 \lVert B \rVert_2^2} \Bigg) 
	\end{align}
	and the number of iterations $R$ satisfy,
	$R \geq \log \big( 2\lVert F^{(1)} - F^* \rVert_2 \epsilon^{-1} \big) / \log(1/T_P)$, (where $T_P$ is defined in Assumption \ref{asm:boundforcontract}) and  
	$\epsilon > 0$. Then, with probability at least $1 - \epsilon^5$, $F^{(r)} \in \FF$, $K^{(1,r)}$ is stabilizing for $r \in [R]$, and
	\begin{align*}
	& \lVert F^{(R)} - F^* \rVert_2 \leq \epsilon, \lVert K^{(S_R,R)} - K^* \rVert_2 \leq (1 + D_0 ) \epsilon, \\
	& \hspace{0.8cm} J(K^{(S_R,R)},F^{(R)}) - J(K^*,F^*) \leq D_1 \epsilon ,
	\end{align*}
	where $D_0$, $D_1$ are absolute constants. 
\end{theorem}

\begin{proof}
	{Let us first split $\lVert F^{(r+1)} - F^* \rVert_2$ using the 
		quantity 
		$\bar{F}^{(r+1)} = \Fs(F^{(r)})$,
		\begin{align} \label{eq:triangle_Fd+1}
		\lVert F^{(r+1)} - F^* \rVert_2 \leq \lVert F^{(r+1)} - \bar{F}^{(r+1)} \rVert_2 + \lVert \bar{F}^{(r+1)} - F^* \rVert_2 .
		\end{align}
		First we bound the term $\lVert \bar{F}^{(r+1)} - F^* \rVert_2 $:
		\begin{align} \label{eq:F_bar_d+1}
		\lVert \bar{F}^{(r+1)} - F^* \rVert_2 = \lVert \Fs(F^{(r)}) - \Fs(F^*) \rVert \leq T_P \lVert F^{(r)} - F^* \rVert_2 .
		\end{align}}
	As for the first term in \eqref{eq:triangle_Fd+1}, we know from the definition of state aggregator \eqref{eq:state_agg} that
	\begin{align*}
	& F^{(r+1)} = A - B (K^{(S_r,r)}_1 + K^{(S_r,r)}_2) , \\ 
	& \bar{F}^{(r+1)} = A - B (K^{(*,r)}_1 + K^{(*,r)}_2) .
	\end{align*}
	{Then,
		\begin{align} \label{eq:F_d+1}
		& \lVert F^{(r+1)} - \bar{F}^{(r+1)} \rVert_2 \nonumber \\
		& \leq \lVert B \rVert_2 \big( \lVert K^{(S_r,r)}_1 - K^{(*,r)}_1 
		\rVert_2 + \lVert K^{(S_r,r)}_2 - K^{(*,r)}_2 \rVert_2 \big) \nonumber\\
		& \leq \sqrt{2}\lVert B \rVert_2 \lVert K^{(S_r,r)} - K^{(*,r)} \rVert_F \nonumber\\
		& \leq \lVert B \rVert_2 \sqrt{ 2 \kappa_3 \epsilon_r } \leq \epsilon   2^{-r-2} .
		\end{align}
		The inequality holds with probability at least $1 - \epsilon^{10}$. The third inequality is due to Proposition \ref{prop:RL-LQG} and the last inequality is due to choice of $\epsilon_r$. Hence, using \eqref{eq:F_bar_d+1} and \eqref{eq:F_d+1},
		\begin{align*}
		\lVert F^{(r+1)} - F^* \rVert_2 \leq T_P \lVert F^{(r)} - F^* \rVert_2 + \epsilon 2^{-r-2},
		\end{align*}
		which holds with probability $1 - \epsilon^{10}$. Furthermore with a union bound argument with $R = \Os(\epsilon^{-5})$, 
		it holds with probability $1-\epsilon^{5}$ that
		\begin{align*}
		\lVert F^{(R)} - F^* \rVert_2 \leq T_P^R \lVert F^{(1)} - F^* \rVert_2 + \epsilon/2 .
		\end{align*}
		As $T_P < 1$ (Assumption \ref{asm:boundforcontract}) and $R$ is defined as in the statement of the theorem, we arrive at 
		\begin{align} \label{eq:final_bound_F}
		\lVert F^{(R)} - F^* \rVert_2 \leq \epsilon ,
		\end{align}
		with probability at least $1 - \epsilon^{5}$.} {Now we prove that 
		$F^{(r)} \in \FF$ for $1 \leq r \leq R$ using recursion. Let 
		$F^{(r)} \in \FF$ for $0 \leq r < R$; then, we can prove that $F^{(r+1)} 
		\in \FF$ with probability at 
		least $1 - \epsilon^{10}$,
		\begin{align*}
		& \lVert F^{(r+1)} \rVert_2 \leq \lVert F^{(r+1)} - \bar{F}^{(r+1)} \rVert_2 + \lVert \bar{F}^{(r+1)} \rVert_2 \\
		& \leq \lVert B \rVert_2 \sqrt{ 2 \kappa_3 \epsilon_r } + T_P \leq (1+T_P)/2 ,
		\end{align*}
		with probability at least $1 - \epsilon^{10}$. The second inequality is 
		due to \eqref{eq:F_d+1} and Lemma \ref{lem:F_contract}. Hence 
		$F^{(r)} \in 
		\FF$ for $1 \leq r \leq R$ with probability at least $1 - 
		\epsilon^{5}$ using a union bound argument.} 
	
	Now we 	prove the second inequality in the statement of the Theorem. Using the 	triangle inequality,
	\begin{align*}
	\lVert K^{(S_R,R)} - K^* \rVert_2 \leq & \lVert K^{(S_R,R)} - K^{(*,R)} 
	\rVert_2 
	+ \\
	& \lVert K^{(*,R)} - K^* \rVert_2 .
	\end{align*}
	From Proposition \ref{prop:RL-LQG}, we know that,
	\begin{align*}
	\lVert K^{(S_R,R)} - K^{(*,R)} \rVert_2 \leq \sqrt{ \kappa_3 \epsilon_R} \leq \epsilon ,
	\end{align*}
	with probability at least $1 - \epsilon^{10}$. Using \eqref{eq:u_t}, \eqref{eq:costate}, and linearity of the mean-field trajectory $\bar{Z}^{(R)}$, 
	\begin{align*}
	K^{(*,R)} = \Bigg[ G_P P A  \hspace{0.3cm} -G_P \sum_{s=0}^{\infty} H_P^s C_Z \big(F^{(R)}\big)^{s+1} \Bigg] .
	\end{align*}
	Similarly, $K^*$ can be defined as,
	\begin{align*}
	K^* = \Bigg[ G_P P A  \hspace{0.3cm} -G_P \sum_{s=0}^{\infty} H_P^s C_Z (F^*)^{s+1} \Bigg] .
	\end{align*}
	Using these definitions,
	\begin{align*}
	& \lVert K^{(*,R)} - K^* \rVert_2 \leq \\
	& \lVert G_P \rVert_2 \sum_{s=0}^{\infty} \lVert H_P \rVert_2^s \lVert C_Z \rVert_2 \lVert \big(F^{(R)}\big)^{s+1} - (F^*)^{s+1} \rVert_2.
	\end{align*}
	Using the matrix manipulations used in \eqref{eq:prop_1_work}, in the proof of Proposition \ref{lem:exist_unique_linear}, 
	we obtain 
	\begin{align*}
	\lVert K^{(*,R)} - K^* \rVert_2 & \leq D_0 \lVert F^{(R)} - F^* \rVert_2 \leq D_0 \epsilon,
	\end{align*}
	with probability at least $1 - \epsilon^{5}$, where 
	\begin{align*}
	D_0 := \frac{\lVert G_P \rVert_2 \lVert C_Z	\rVert_2 }{(1 - \lVert H_P \rVert_2)^2}.
	\end{align*}
	The last inequality is due to \eqref{eq:final_bound_F}, and thus we obtain,
	\begin{align} \label{eq:final_bound_K}
	\lVert K^{(S_R,R)} - K^* \rVert_2 \leq (1 + D_0 ) \epsilon ,
	\end{align}
	with probability at least $1 - \epsilon^{5}$. {Now we prove that $K^{(1,r+1)}$ is a stabilizing controller for system $(\bar{A}^{(r+1)}, \bar{B})$, given that $K^{(1,r)}$ is a stabilizing controller for system $(\bar{A}^{(r)}, \bar{B})$ for $0 \leq r < R$}. {We know from \cite{yang2019global} that for any stabilizing $K$,
	\begin{align} \label{eq:def_J}
	J(K,F) = \tr \big( \big(C_X + K^T C_U K \big) \Sigma_K \big) + \sigma^2 \tr(C_U) ,
	\end{align}
	where $\Sigma_K$ is the covariance matrix for the stationary distribution induced by $K$ and is the unique positive definite solution to the Lyapunov equation,}
	\begin{align*}
	\Sigma_K = \Sigma_X + \big(\bar{A} - \bar{B} K \big) \Sigma_K \big(\bar{A} - \bar{B} K \big)^T.
	\end{align*}
	Hence $\Sigma_K$ satisfies,
	\begin{align} \label{eq:def_sigma_k}
	\Sigma_K & = \sum_{t=0}^{\infty} \big(\bar{A} - \bar{B} K \big)^t \Sigma_X \big(\big(\bar{A} - \bar{B} K \big)^t\big)^T \nonumber \\
	& = \sum_{t=0}^{\infty} \big(A - B K_1 \big)^t \Sigma^{11}_X \big(\big(A - B K_1 \big)^t\big)^T , 
	\end{align}
	where $\Sigma^{11}_X = \Sigma_{w} + \sigma^2 B B^T$. From Proposition \ref{prop:RL-LQG} we know that since $K^{(1,r)}$ is stabilizing controller for system $(\bar{A}^{(r)},\bar{B})$, $K^{(S_r,r)}$ is also stabilizing, and hence $J(K^{(S_r,r)}, F^{(r)})$ is finite. Also from Algorithm \ref{alg:actorcritic} we know that $K^{(1,r+1)} = K^{(S_r,r)}$. Using \eqref{eq:def_J} and \eqref{eq:def_sigma_k} we can 	deduce that $J(K^{(1,r+1)}, F^{(r+1)}) = J(K^{(1,r+1)}, F^{(r)})$ as long as 	$F^{(r+1)}$ is stable. This is true with probability at least $1 - \epsilon^{10}$. Finally $J(K^{(1,r+1)}, F^{(r)}) = J(K^{(S_r,r)}, F^{(r)}) < \infty$, and hence $K^{(1,r+1)}$ is a stabilizing controller for system $(\bar{A}	^{(r+1)}, \bar{B})$ for $0 \leq r < R$ with probability at least $1 - \epsilon^{5}$ using a union bound argument.
	
	Now we prove the third inequality from the theorem. For the rest of the proof we will introduce $K^{(R)} := K^{(S_R,R)}$ for conciseness. We begin by obtaining an expression for the cost of a stabilizing controller $K$. 
	Since $K^{(R)}$ and $K^*$ are stabilizing controllers, there exists some $\rho_1 \in (0,1)$ and $c > 0$ s.t. 
	\begin{align*}
	\lVert (\bar{A} - \bar{B} K^{(R)})^t \rVert_2 \leq c \rho^t_1,\hspace{0.2cm} \lVert (\bar{A} - \bar{B} K^*)^t \rVert_2 \leq c \rho^t_1.
	\end{align*}
	Using \eqref{eq:def_J} and \eqref{eq:def_sigma_k}, we can write $J(K^{(R)},F^{(R)})$ and $	J(K^*,F^*)$ 
	\begin{align} \label{eq:cost_bound}
	& J(K^{(R)},F^{(R)}) = \sigma^2 \tr(C_U) + \nonumber \\
	& \tr \Big( \big( C_X + K^{(R)} C_U \big( K^{(R)} \big)^T \big)\sum_{t=0}^{\infty} L_{R}^t \Sigma^{11}_X \big( L_{R}^t \big)^T \Big) ,\nonumber \\
	& J(K^{*},F^{*}) = \sigma^2 \tr(C_U) + \nonumber \\
	& \tr \Big( \big( C_X + K^{*} C_U \big( K^{*} \big)^T \big) \sum_{t=0}^{\infty} L_{*}^t \Sigma^{11}_X \big( L_{*}^t \big)^T \Big) ,
	\end{align}
	where we have introduced 
	\begin{align*}
	L_{R}:= A - B K^{(R)}_1, \hspace{0.2cm }L_{*} := A - B K^{*}_1.
	\end{align*}
	Using \eqref{eq:augdynamics}-\eqref{eq:augdynamics_2} and the fact that $K^{(R)} = [K^{(R)}_1,K^{(R)}_2]$ and $K^* = [K^*_1,K^*_2]$, we can deduce that $\lVert L^t_{R} \rVert_2 \leq c \rho^t_1$ and $\lVert L^t_* \rVert_2 \leq c \rho^t_1$. Now taking the difference of both terms in equation \eqref{eq:cost_bound},
	\begin{align} \label{eq:fin_sample_cost}
	& J(K^{(R)},F^{(R)}) - J(K^{*},F^{*})  \nonumber \\
	& = \tr \Big( \big( C_X + K^{(R)} C_U \big( K^{(R)} \big)^T 	\big)\sum_{t=0}^{\infty} \big(L_{R}^t \Sigma^{11}_X \big( L_{R}^t \big)^T \Big) 	\nonumber\\
	& \hspace{0.4cm} - \tr \Big( \big( C_X + K^{*} C_U \big( K^{*} \big)^T \big) \sum_{t=0}^{\infty} L_{*}^t \Sigma^{11}_X \big( L_{*}^t \big)^T \Big) , \nonumber \\
	& = \tr \Big( \big( C_X + K^{*} C_U \big( K^{*} \big)^T \big) \nonumber \\
	& \sum_{t=0}^{\infty} \big(L_{R}^t \Sigma^{11}_X \big( L_{R}^t \big)^T - L_{*}^t \Sigma^{11}_X \big( L_{*}^t \big)^T \big) \Big) + \nonumber \\
	& \tr\Big( \big( K^{*} C_U \big( K^{*} \big)^T - K^{(R)} C_U \big( K^{(R)} \big)^T \big) \sum_{t=0}^{\infty} L_{*}^t \Sigma^{11}_X \big( L_{*}^t \big)^T \Big).
	\end{align}
	The first term on the RHS of \eqref{eq:fin_sample_cost} can be upper bounded using the fact that $C_X + K^{*} C_U \big( K^{*} \big)^T$ is symmetric: 
	\begin{align}
	& \big\lVert \big( C_X + K^{*} C_U \big( K^{*} \big)^T \big) \big\rVert_F \nonumber \\
	& \tr \Big(\sum_{t=0}^{\infty} \big(L_{R}^t \Sigma^{11}_X \big( L_{R}^t \big)^T - L_{*}^t \Sigma^{11}_X \big( L_{*}^t \big)^T \big) \Big) \nonumber \\
	& \hspace{2cm} \leq m^{3/2} \big\lVert \big( C_X + K^{*} C_U \big( K^{*} \big)^T \big) \big\rVert_F \nonumber \\
	& \hspace{2cm} \sum_{t=0}^{\infty} \Big\lVert L_{R}^t \Sigma^{11}_X \big( L_{R}^t \big)^T - L_{*}^t \Sigma^{11}_X \big( L_{*}^t \big)^T \Big\rVert_2 .
	\end{align}
	This inequality is due to the fact that for any square matrix $A$ with 	dimension $m \times m$, 
	\begin{align*}
	\tr(A) & \leq \sum_{i,j} \lvert A_{i,j} \rvert \leq m \sqrt{\sum_{i,j} (A_{i,j})^2 } \\
	& = m \lVert A \rVert_F \leq m^{3/2} \lVert A \rVert_2.
	\end{align*} 
	Now we bound the following term 
	\begin{align} 
	& \sum_{t=0}^{\infty} \Big\lVert L_{R}^t \Sigma^{11}_X \big( L_{R}^t \big)^T - L_{*}^t \Sigma^{11}_X \big( L_{*}^t \big)^T \Big\rVert_2 \nonumber \\
	& = \sum_{t=0}^{\infty} \Big\lVert L_{R}^t \Sigma^{11}_X \big( L_{R}^t - L_{*}^t \big)^T + \big( L_{R}^t -  L_{*}^t \big) \Sigma^{11}_X \big( L_{*}^t \big)^T \Big\rVert_2 \nonumber \\
	& \leq \big\lVert \Sigma^{11}_X \big\rVert_2 \sum_{t=0}^{\infty} \big( 		\big\lVert L_{R}^t \big\rVert_2 + \big\lVert L_{*}^t \big\rVert_2 \big)		\big\lVert L_{R}^t - L_{*}^t \big\rVert_2 \nonumber 
	\end{align}
	\begin{align}\label{eq:cost_term_1}
	& \leq 2c \big\lVert \Sigma^{11}_X \big\rVert_2 \sum_{t=0}^{\infty} \rho_1^t
	\big\lVert L_{R}^t - L_{*}^t \big\rVert_2 \nonumber \\
	& = 2c\big\lVert \Sigma^{11}_X \big\rVert_2 \sum_{t=0}^{\infty}  \rho_1^t \bigg\lVert \sum_{\tau = 0}^{t-1} L_{R}^{t-\tau-1} \big(L_{R} - L_{*} \big) L_{*}^{\tau} \bigg\rVert_2 \nonumber \\
	& \leq 2c^3 \big\lVert \Sigma^{11}_X \big\rVert_2 \sum_{t=0}^{\infty}  t \rho_1^{t} \big\lVert L_{R} - L_{*} \big\rVert_2 \nonumber \\
	& \leq 2c^3 \frac{\big\lVert \Sigma^{11}_X \big\rVert_2 \big\lVert B \big\rVert_2 \rho_1}{ (1 - \rho_1)^2}  \big\lVert K^{(R)} - K^* \big\rVert_2 \nonumber \\
	& \leq 2 c^3 \frac{\big\lVert \Sigma^{11}_X \big\rVert_2 \big\lVert B \big\rVert_2 \rho_1}{ (1 - \rho_1)^2}  \big\lVert K^{(R)} - K^* \big\rVert_F .
	\end{align}
	The second equality is obtained by following a procedure similar to \eqref{eq:prop_1_work} in the proof of Proposition \ref{lem:exist_unique_linear}. Now we upper bound the second term in \eqref{eq:fin_sample_cost}:
	\begin{align} \label{eq:cost_term_2}
	& \hspace{-0.3cm} \tr\Big( \big( K^{*} C_U \big( K^{*} \big)^T - K^{(R)} C_U 
	\big( K^{(R)} \big)^T \big) \sum_{t=0}^{\infty} L_{*}^t \Sigma^{11}_X 
	\big( L_{*}^t \big)^T \Big) \nonumber \\
	& \leq \big\lVert K^{*} C_U \big( K^{*} \big)^T - K^{(R)} C_U \big( K^{(R)} 
	\big)^T \big\rVert_F \nonumber \\
	& \hspace{4cm} \tr\Big(\sum_{t=0}^{\infty} L_{*}^t \Sigma^{11}_X 
	\big( L_{*}^t \big)^T \Big) .
	\end{align}
	The inequality is due to the fact that $K^{*} C_U \big( K^{*} \big)^T - K^{(R)} C_U \big( K^{(R)} \big)^T$ is symmetric. The summation in \eqref{eq:cost_term_2} can be upper bounded as follows:
	\begin{align} \label{eq:cost_term_3}
	& \tr\Big(\sum_{t=0}^{\infty} L_{*}^t \Sigma^{11}_X \big( L_{*}^t \big)^T \Big) = \sum_{t=0}^{\infty} \lVert L_{*}^t (\Sigma^{11}_X)^{\frac{1}{2}} \rVert_F^2 \nonumber \\
	& \leq m \sum_{t=0}^{\infty} \lVert L_{*}^t (\Sigma^{11}_X)^{\frac{1}{2}} \rVert_2^2 \leq \frac{m c^2 \lVert (\Sigma^{11}_X)^{\frac{1}{2}} \rVert_2^2}{1 - \rho_1^2} .
	\end{align}
	The first term in \eqref{eq:cost_term_2} can be upper bounded as follows,
	\begin{align} \label{eq:cost_term_4}
	& \big\lVert K^{(R)} C_U \big( K^{(R)} 
	\big)^T - K^{*} C_U \big( K^{*} \big)^T \big\rVert_F \nonumber \\
	& = \big\lVert \big(K^{(R)} - K^* + K^* \big) C_U \big( K^{(R)} - K^* + K^*
	\big)^T \nonumber\\
	& \hspace{0.4cm} - K^{*} C_U \big( K^{*} \big)^T \big\rVert_F \nonumber\\
	& = \big\lVert \big(K^{(R)} - K^* \big) C_U \big( K^{(R)} - K^*\big)^T + 
	\nonumber\\
	& \hspace{0.4cm}  2 \big(K^{(R)} - K^* \big) C_U \big( K^* \big)^T 
	\big\rVert_F 
	\nonumber\\
	& \leq \big\lVert C_U \big\rVert_F \Big[ \big\lVert K^{(R)} - K^* 
	\big\rVert_F^2 
	+ 2 \lVert K^* \rVert_F \big\lVert K^{(R)} - K^* \big\rVert_F \Big] .
	\end{align}
	Hence, the difference in costs $J(K^{(R)},F^{(R)}) - J
	(K^{*},F^{*})$ 
	using \eqref{eq:fin_sample_cost}-\eqref{eq:cost_term_4} is
	\begin{align} \label{eq:cost_term_5}
	& J(K^{(R)},F^{(R)}) - J(K^{*},F^{*}) \leq \nonumber \\
	& \hspace{1.5cm}  D_2 \big\lVert K^{(R)} - K^* \big\rVert_F + D_3 \big\lVert 
	K^{(R)} - K^* \big\rVert_F^2 ,
	\end{align}
	where $D_2$ and $D_3$ are defined as follows:
	\begin{align} \label{eq:cost_term_6}
	D_2 := & 2m^{3/2} c^3 \big\lVert \big( C_X + K^{*} C_U \big( K^{*} \big)^T \big) \big\rVert_F \frac{\big\lVert \Sigma^{11}_X \big\rVert_2 \big\lVert B \big\rVert_2 \rho_1}{ (1 - \rho_1)^2} \nonumber \\
	& + 2 m c^2 \lVert K^* \rVert_F \big\lVert C_U \big\rVert_F \frac{ \lVert (\Sigma^{11}_X)^{\frac{1}{2}} \rVert_2^2}{1 - \rho_1^2} , \nonumber \\
	D_3 := & \big\lVert C_U \big\rVert_F \frac{m c^2 \lVert (\Sigma^{11}_X)^{\frac{1}{2}} \rVert_2^2}{1 - \rho_1^2} .
	\end{align}
	Using \eqref{eq:final_bound_K}, \eqref{eq:cost_term_5}-\eqref{eq:cost_term_6} 
	for $\epsilon$ small enough, with probability at 
	least $1 
	- \epsilon^{5}$,
	\begin{align*}
	& J(K^{(R)},F^{(R)}) - J(K^{*},F^{*}) \leq D_1 \epsilon,
	\end{align*}
	where $D_1 := D_2 \big(1 + D_0 \big) + D_3 \big(1 + D_0 \big)^2$.
\end{proof}

\subsection{Approximate $\epsilon$-NE bound}
We now quantify how the approximate MFE obtained from Theorem 
\ref{thm:conv_guarantee} performs in the original finite population game. Let us denote the control law generated by the approximate MFE $(F^{(R)}, K^{(S_R,R)})$ in Algorithm \ref{alg:actorcritic} by $\tilde{\mu}$. The approximate MFE controller for agent $n$ is 
\begin{align}\label{eq:eps_Nash_control}
\tilde{\mu}_t^n(Z^n_t) = - K^{(S_R,R)}_1 Z^n_t - K^{(S_R,R)}_2 (F^{(R)})^t \nu_0
\end{align}
and $\tilde{\mu}^{-n}$ is the joint policy of all agents except agent $n$, $\tilde{\mu}^{-n} = (\tilde\mu^1, \ldots,\tilde\mu^{n-1},\tilde\mu^{n+1},\ldots,\tilde\mu^{N})$.


\begin{theorem} \label{thm:fin_agent_bound}
Let the output cost of Algorithm \ref{alg:actorcritic} for a finite population LQ game for agent $n$ be $J^N_n (\tilde{\mu}^n, \tilde{\mu}^{-n})$, and denote the NE cost of this game by $\inf_{\pi^n \in \Pi} J^N_n(\pi^n,\tilde{\mu}^{-n})$. Then, 
if $R = \Omega(\log(1/\epsilon))$ and $N = \Omega(1/\epsilon^2)$,
\begin{align*}
J^N_n (\tilde{\mu}^n, \tilde{\mu}^{-n}) - \inf_{\pi^n \in \Pi} J^N_n(\pi^n,\tilde{\mu}^{-n}) \leq \epsilon ,
\end{align*}
with probability at least $1-\epsilon^5$. 
\end{theorem}
\begin{proof}
	The quantity $J^N_n (\tilde{\mu}^n, \tilde{\mu}^{-n}) - \inf_{\pi^n \in \Pi} 
	J^N_n(\pi^n, \tilde{\mu}^{-n})$ can be broken up into two terms,
	\begin{align} \label{eq:exp1_thm11}
	& J^N_n (\tilde{\mu}^n, \tilde{\mu}^{-n}) - \inf_{\pi^n \in \Pi} 
	J^N_n(\pi^n, \tilde{\mu}^{-n}) = \nonumber \\
	& \hspace{1cm} J^N_n (\tilde{\mu}^n, \tilde{\mu}^{-n}) - J(K^{(S_R,R)},F^{(R)}) + \nonumber \\
	& \hspace{1cm} J(K^{(S_R,R)},F^{(R)}) - \inf_{\pi^n \in \Pi} J^N_n(\pi^n, \tilde{\mu}^{-n}).
	\end{align}
	For simplicity let us denote by $\bar{Z} = (\bar{Z}_0, \bar{Z}_1, \ldots)$ the mean-field trajectory consistent with the control law $\tilde\mu$. Similarly let $\bar{Z}^N_{n}= (\bar{Z}^N_{n,0}, \bar{Z}^N_{n,1}, \ldots)$ be is the empirical mean-field trajectory under $\tilde{\mu}$ such that,
	\begin{align*}
	\bar{Z}^N_{n,t} = \frac{1}{N-1} \sum_{n' \neq n} Z^{n'}_t, 
	\end{align*} 
	where $Z^{n} = (Z^{n}_0,Z^{n}_1,\ldots)$ is the trajectory of agent $n$ under control law $\tilde \mu$. 
	Using Lemma 3 in \cite{moon2014discrete}, since it is applicable for any stabilizing controller,
	\begin{align} \label{eq:exp1_thm21}
	& J^N_n (\tilde{\mu}^n, \tilde{\mu}^{-n}) - J(K^{(S_R,R)},F^{(R)}) = \nonumber \\
	& \hspace{1cm} \Os \Bigg( \sqrt{\limsup_{T \rightarrow \infty} \EE \big( \sum_{t=0}^{T-1} \big\lVert \bar{Z}^{N}_{n,t} - \bar{Z}_t \big\rVert_2^2 \big)/T} \Bigg) .
	\end{align}
	Now we obtain an upper bound for the expression on the RHS. Using \eqref{eq:finitesystem} the dynamics of $\bar{Z}^{N}_{n}$ can be written as
	\begin{align*}
	\bar{Z}^{N}_{n,t+1} = (A - B K_1^{(S_R,R)}) \bar{Z}^{N}_{n,t} - B K_2^{(S_R,R)} \bar{Z}_t + W^{N}_{n,t},
	\end{align*}
	where 
	\begin{align*}
	W^{N}_{n,t} := \sum_{n' \neq n} (W^{n'}_t + B \zeta^{n'})/(N-1),
	\end{align*} 
	is generated i.i.d. with distribution $\Ns(0,\Sigma^N_w)$ where 
	\begin{align*}
	\Sigma^N_w = \frac{\Sigma_w + \sigma^2 B B^T}{N-1}
	\end{align*} 
	and $\EE [\bar{Z}^{N}_{n,0}] = \nu_0$. If we augment $\bar{Z}^{N}_{n,t}$ with $\bar{Z}_t$, the state dynamics of this augmented system can be written as
	\begin{align*}
	\left( \begin{array}{c}
	\bar{Z}^{N}_{n,t+1} \\ \bar{Z}_{t+1}\end{array} \right) = \big(\bar{A} - \bar{B} K^{(S_R,R)}\big) \left( \begin{array}{c}\bar{Z}^{N}_{n,t} \\ \bar{Z}_{t} \end{array} \right) + \left( \begin{array}{c}W^{N}_{n,t} \\ 0\end{array} \right)
	\end{align*}
	and the expression inside the parentheses on the RHS of \eqref{eq:exp1_thm21} can be expressed as the cost function: 
	\begin{align*} 
	& \limsup_{T \rightarrow \infty} \frac{1}{T} \EE \bigg( \sum_{t=0}^{T-1} \bigg\lVert \begin{array}{c}\bar{Z}^{N}_{n,t} \\ \bar{Z}_{t}\end{array} \bigg\rVert_D^2 \bigg) \text{ where } D = \left( \begin{array}{cc} I & -I \\ -I & I \end{array} \right) .
	\end{align*}
	Since $K^{(S_R,R)}$ is a stabilizing controller for system $(\bar{A}, \bar{B})$, using Proposition 3.1 from \cite{yang2019global}, we know that the cost in expression above is $\Os(\tr(\Sigma^N_w)) = \Os(1/(N-1))$. Hence,
	\begin{align} \label{eq:bound_finite_vs_infty1}
	\sqrt{\limsup_{T \rightarrow \infty} \EE \big( \sum_{t=0}^{T-1} \big\lVert \bar{Z}^{N}_{n,t} - \bar{Z}_t \big\rVert_2^2 \big)/T} = \Os \bigg( \frac{1}{\sqrt{N-1}} \bigg).
	\end{align}
	
	Next we use the approach of Lemma 4 in \cite{moon2014discrete} to obtain an upper bound on 
	\begin{align*}
	J(K^{(S_R,R)},F^{(R)}) - J^N_n(\pi^n, \tilde{\mu}^{-n}),
	\end{align*}
	for any $\pi^n \in \Pi$. Denote the trajectory of agent $n$ under control law $\pi^n$ by $\check{Z}^n_t$. The cost for agent $n$ under control law $\pi^n$, while the other agents in the $N$-agent game are following the approximate MFE controller $\tilde{\mu}$, is
	\begin{align} \label{eq:exp2_thm21}
	& J^N_n(\pi^n,\tilde{\mu}^{-n}) \nonumber \\
	& = \limsup_{T \rightarrow \infty} \frac{1}{T} \sum^{T-1}_{t=0} \EE \big[ \lVert \check{Z}^n_t - \bar{Z}_t + \bar{Z}_t - \bar{Z}^N_{n,t} \rVert^2_{C_Z} + \lVert U_t \rVert^2_{C_X} \big] \nonumber \\
	& \geq \limsup_{T \rightarrow \infty} \frac{1}{T} \sum^{T-1}_{t=0} \EE \big[ \lVert \check{Z}^n_t - \bar{Z}_t \rVert^2_{C_Z} + \lVert U_t \rVert^2_{C_X} \big] \nonumber \\
	& \hspace{1.0cm}  + \limsup_{T \rightarrow \infty} \frac{2}{T} \sum_{t=0}^{T-1} \EE \big[ (\check{Z}^n_t - \bar{Z}_t)^T C_Z (\bar{Z}_t - \bar{Z}^N_{n,t}) \big] \nonumber \\
	& \geq J(\pi^n,\bar{Z}) + \limsup_{T \rightarrow \infty} \frac{2}{T} \sum_{t=0}^{T-1} \EE \big[ (\check{Z}^n_t - \bar{Z}_t)^T C_Z (\bar{Z}_t - \bar{Z}^N_{n,t}) \big] \nonumber \\
	& \geq J(K^{(S_R,R)},F^{(R)}) - 2 D_1 \lVert F^{(1)} - F^* \rVert_2 T^R_P + \nonumber \\
	& \hspace{1.0cm} \limsup_{T \rightarrow \infty} \frac{2}{T} \sum_{t=0}^{T-1} \EE \big[ (\check{Z}^n_t - \bar{Z}_t)^T C_Z (\bar{Z}_t - \bar{Z}^N_{n,t} )\big],
	\end{align}
	with probability at least $1-\epsilon^5$. 
	The last inequality is obtained by using the definition of $R$ and the upper bound on $J(K^{(S_R,R)},F^{(R)}) - J(K^*,F^*)$ from Theorem \ref{thm:conv_guarantee}. The last term on the RHS of \eqref{eq:exp2_thm21} can be bounded by techniques used in Lemma 4 in \cite{moon2014discrete}:
	\begin{align} \label{eq:exp3_thm21}
	& \bigg\lvert \limsup_{T \rightarrow \infty} \frac{2}{T} \sum_{t=0}^{T-1} \EE \big[ (\check{Z}^n_t - \bar{Z}_t)^T C_Z (\bar{Z}_t - \bar{Z}^N_{n,t} \big] \bigg\rvert = \nonumber \\
	& \hspace{1.5cm} \Os \Bigg( \sqrt{\limsup_{T \rightarrow \infty} \EE \big( \sum_{t=0}^{T-1} \big\lVert \bar{Z}^{N}_{n,t} - \bar{Z}_t \big\rVert_2^2 \big)/T} \Bigg).
	\end{align}
	Using \eqref{eq:bound_finite_vs_infty1} -- \eqref{eq:exp3_thm21}, 
	\vspace{-0.15cm}\begin{align} \label{eq:exp4thm21}
	& J(K^{(S_R,R)},F^{(R)}) - J^N_n(\pi^n,\tilde{\mu}^{-n}) = \nonumber \\
	& \hspace{3.5cm} \Os \big( T^R_P \big) + \Os \Big( \frac{1}{\sqrt{N-1}} \Big), 
	\end{align}
	with probability at least $1-\epsilon^5$. Now using \eqref{eq:exp1_thm11}, \eqref{eq:exp1_thm21}, \eqref{eq:bound_finite_vs_infty1} and \eqref{eq:exp4thm21}, we arrive at
	\begin{align*}
	J^N_n (\tilde{\mu}^n, \tilde{\mu}^{-n}) - J^N_n(\pi^n,\tilde{\mu}^{-n})= \Os \big( T^R_P \big) + \Os \Big( \frac{1}{\sqrt{N-1}} \Big),
	\end{align*}
	with probability at least $1-\epsilon^5$. Hence, if $R = \Omega(\log(1/\epsilon))$ and $N = \Omega(1/\epsilon^2)$, then, with probability at least $1-\epsilon^5$,
	\begin{align*}
	J^N_n (\tilde{\mu}^n, \tilde{\mu}^{-n}) - J^N_n(\pi^n,\tilde{\mu}^{-n}) \leq \epsilon ,
	\end{align*}
	for any $\pi^n \in \Pi$ and this concludes the proof.
	
\end{proof}

\section{Concluding Remarks}
\label{sec:Conc}

This paper has proposed an RL algorithm to find the approximate MFE of the non-stationary LQ-MFG 
in a model-free setting. This is achieved by, 1) reformulating the LQT problem 
into a forward-in-time problem, and 2) generalizing actor-critic to the unmixed MC setting. 
Furthermore, the proposed learning algorithm yields an approximate MFE, which is shown to be an $\epsilon$-NE, dependent upon both the number of agents in the finite population game and the number of iterations of the learning algorithm. 

\bibliographystyle{IEEEtran} 
\bibliography{references,MARL_Springer_1,MARL_Springer_2,RL}

\end{document}